\documentclass[journal,twoside,web]{ieeecolor}
\usepackage{generic}
\usepackage{cite}
\usepackage[hidelinks]{hyperref}
\usepackage{amsmath,amssymb,amsfonts}
\usepackage{algpseudocode,algorithm}
\usepackage{graphicx}
\usepackage{textcomp}
\usepackage{subcaption}
\usepackage{xcolor}
\usepackage{tabularx}
\usepackage{academicons}
\usepackage{multirow}
\usepackage{diagbox}
\usepackage[colorinlistoftodos]{todonotes}

\def\Ac{{\mathcal A}}
\def\Abb{{\mathbb A}}

\def\Cc{{\mathcal C}}
\def\Cbb{{\mathbb C}}

\def\cfr{{\mathfrak c}}

\def\Dc{{\mathcal D}}
\def\Dbb{{\mathbb D}}

\def\Ec{{\mathcal E}}
\def\Ebb{{\mathbb E}}

\def\Gc{{\mathcal G}}

\def\Hc{{\mathcal H}}

\def\Lc{{\mathcal L}}

\def\Mc{{\mathcal M}}

\def\Nc{{\mathcal N}}

\def\Oc{{\mathcal O}}

\def\Rbb{{\mathbb R}}

\def\Sbb{{\mathbb S}}

\def\Tc{{\mathcal T}}

\def\Vc{{\mathcal V}}

\def\0{{\bf 0}}

\newcommand{\bitem}{\begin{itemize}}
\newcommand{\eitem}{\end{itemize}}
\newcommand{\btabular}{\begin{tabular}}
\newcommand{\etabular}{\end{tabular}}
\newcommand{\bcenter}{\begin{center}}
\newcommand{\ecenter}{\end{center}}
\newcommand{\bea}{\begin{eqnarray}}
\newcommand{\eea}{\end{eqnarray}}
\newcommand{\bean}{\begin{eqnarray*}}
\newcommand{\eean}{\end{eqnarray*}}
\newcommand{\ba}{\left[ \begin{array}}
\newcommand{\ea}{\\ \end{array} \right]}
\newcommand{\bear}{\begin{array}}
\newcommand{\eear}{\\ \end{array}}

\newcommand{\non}{\nonumber}

\newcommand{\ra}{\rightarrow}

\newcommand*{\QET}{\hfill\ensuremath{\triangleleft}}
\newcommand{\norm}[1]{\left\lVert#1\right\rVert}

\newcounter{subequation}
\def\beasub{\addtocounter{equation}{+1}
\setcounter{subequation}{\value{equation}}
\setcounter{equation}{0}
\renewcommand{\theequation}{\arabic{subequation}\alph{equation}}
\begin{eqnarray}}
\def\eeasub{\end{eqnarray}
\setcounter{equation}{\value{subequation}}
\renewcommand{\theequation}{\arabic{equation}}}

\def\inf{\operatornamewithlimits{inf\vphantom{p}}}

\newtheorem{Lemma}{Lemma}
\newtheorem{Theorem}{Theorem}
\newtheorem{Definition}{Definition}
\newtheorem{Assumption}{Assumption}
\newtheorem{Remark}{Remark}
\newtheorem{Problem}{Problem}
\newtheorem{Proposition}{Proposition}

\makeatletter
\newcommand{\multiline}[1]{%
  \begin{tabularx}{\linewidth-\ALG@thistlm-0.0cm}[t]{@{}X@{}}
    #1
  \end{tabularx}
}
\makeatother
\def\BibTeX{{\rm B\kern-.05em{\sc i\kern-.025em b}\kern-.08em
T\kern-.1667em\lower.7ex\hbox{E}\kern-.125emX}}
\begin{document}
\title{Security Allocation in Networked Control Systems under Stealthy Attacks
}
\author{Anh Tung Nguyen, André M. H. Teixeira, Alexander Medvedev 
\thanks{ This work is supported by the Swedish Research Council under the
grants 2018-04396 and 2021-06316 and by the Swedish Foundation for Strategic Research.}
\thanks{Anh Tung Nguyen, Andr{\'e} M. H. Teixeira, and Alexander Medvedev are with the Department of Information Technology, Uppsala University, PO Box 337, SE-75105, Uppsala, Sweden (e-mail: \{anh.tung.nguyen, andre.teixeira, alexander.medvedev\}@it.uu.se).}
}
	
\maketitle
	
\begin{abstract}
This paper considers the problem of security allocation in a networked control system under stealthy attacks. The system is comprised of interconnected subsystems represented by vertices.
A malicious adversary selects a single vertex on which to conduct a stealthy data injection attack {with the purpose of} maximally {disrupting a distant target vertex} while remaining undetected.
Defense resources against the adversary are allocated by 
a defender
on several selected vertices.
First, the objectives of the adversary and the defender with uncertain targets are formulated in a probabilistic manner, resulting in an expected worst-case impact of stealthy attacks.
Next, we provide a graph-theoretic necessary and sufficient condition under which the cost for the defender and the expected worst-case impact of stealthy attacks are bounded.
This condition enables the defender to restrict the admissible actions to 
{dominating sets of the graph representing the network.}
Then, {the security allocation problem is solved through}
a Stackelberg game-theoretic framework. 
Finally, {the obtained results are validated through}
a numerical example of a 50-vertex networked control system.
\end{abstract}
	
\section{Introduction}
\label{sec:introduction}
Networked control systems are ubiquitous in modern society and are exemplified by power grids, transportation, and water distribution networks.
These systems, utilizing non-proprietary information and communication technologies, such as public Internet and wireless communication, are exposed to the threat of cyber attacks \cite{teixeira2015secure,falliere2011w32,kshetri2017hacking}, 
with potentially severe financial and societal consequences.
For instance, an Iranian industrial control system and a Ukrainian power grid have witnessed the catastrophic consequences of malware such as Stuxnet in 2010 \cite{falliere2011w32} and Industroyer in 2016 \cite{kshetri2017hacking}, respectively.
Thus, in light of these alarming realities, the issue of security has acquired unprecedented significance in the realm of control systems. 

In terms of cyber attacks on control systems, 
deception attacks that undermine the integrity of control systems have emerged as an area of increasing scholarly interest. 
For example, Pang and Liu \cite{pang2011design} have proposed an encryption-based predictive control mechanism to counteract and mitigate such attacks. 
Another form of deception attacks, replay attacks, has been unmasked by physical watermarking \cite{mo2009secure,naha2023quickest} and multiplicative watermarking \cite{ferrari2020switching}.
Meanwhile, the development of stealthy attacks on control systems has been made to evade the most advanced detection schemes \cite{park2019stealthy,zhang2020false,ren2021kullback,li2023secure}.

Upon review of the above 
studies \cite{park2019stealthy,zhang2020false,ren2021kullback,li2023secure,mo2009secure,pang2011design,naha2023quickest,ferrari2020switching}, it is noticed that they have concentrated on secure estimation and secure control from the perspective of either the defender or the adversary. 
Nonetheless, it is crucial to note that both parties are confronted with similar challenges, as the defender has limited resources to counteract malicious activities, while the adversary also faces energy and detectability constraints when executing attacks. 
As a result, 
addressing the security problem within a unified framework that encompasses both the defender and the adversary is of utmost importance.

Game theory offers a unified framework to consider the objectives and actions of both strategic players, namely the defender and the adversary \cite{zhu2015game}.
It also allows us to deal with the robustness and security of cyber-physical systems within the common well-defined framework of $\Hc_\infty$ robust control design \cite{wang2007lmi}.
Further, many other concepts of games describing networked systems subjected to cyber attacks such as matrix games \cite{nguyen2022single,umsonst2021bayesian,shukla2022robust}, dynamic games \cite{gupta2016dynamic}, stochastic games \cite{miao2018hybrid},
and network monitoring games \cite{milovsevic2023strategic} 
have been 
studied.
Recent studies \cite{wu2020zero,nguyen2022single,nguyen2022zero} have utilized the common concept of zero-sum games to address the problem of input attacks on cyber-physical systems.
Control systems exposed to cyber attacks have been extensively investigated through game theoretic approaches \cite{gupta2016dynamic,miao2018hybrid,milovsevic2023strategic}. 
However, these approaches have not accounted for the deployment of detectors in an effort to improve the detection of cyber attacks. 
This creates a significant knowledge gap that must be addressed in order to enhance security measures.

One such effort to close the aforementioned gap has been presented in a game-theoretic formulation outlined by Pirani et al. \cite{pirani2021game}. 
The game payoff in \cite{pirani2021game} has been formulated by combining the maximum $\mathcal{L}_2$ gains of multiple outputs with respect to a single input representing the attack signal.
On the one hand, these multiple $\mathcal{L}_2$ gains are evaluated separately and thus may be attained for different 
input signals. 
Further, the utilization of a maximum gain for characterizing the detectability corresponds to an optimistic perspective, where the adversary attempts to maximize the energy of the detection output, instead of the opposite.
Therefore, in order to address the critical issue of cyber security and develop a security metric against cyber attacks, it is imperative to thoroughly investigate the optimal placement of sensors in a networked system to minimize the impact of cyber attacks while maintaining maximum detectability.

Additionally, the above existing studies \cite{nguyen2022single,nguyen2022zero,pirani2021game,umsonst2021bayesian,gupta2016dynamic,milovsevic2023strategic} considered the security problem where the defender and the adversary select their actions simultaneously.
However, this formulation is not always applicable in practical situations where an adversary attacks after observing the action of the defender.
To deal with this scenario,
a game-theoretic Stackelberg framework \cite{bacsar1998dynamic} offers a more practical solution \cite{shukla2022robust,li2018false,yuan2019stackelberg}.
In the framework, after analyzing possible attack scenarios, the defender called \textit{the leader}, has the power to select and announce their action first, knowing that the malicious adversary bases their actions on the leader's decision.
Then, the malicious adversary called \textit{the follower}, finds the best response to the defender's action.

In this paper, we consider a continuous-time networked control system, associated with an undirected connected graph, under stealthy attacks involving two strategic agents: a malicious adversary and a defender.
The system is comprised of multiple interconnected one-dimensional subsystems, referred to as vertices.
The purpose of the adversary is to maliciously degrade 
{a distant target vertex}
without being detected. 
To 
this end,
the adversary selects one vertex on which to launch a stealthy data injection attack on its input. 
Meanwhile, the defender allocates defense resources by selecting a set of monitor vertices
to measure their outputs with the aim of alleviating the attack impact. 
Given the strategic nature of both agents, we investigate the optimal selection of the monitor vertices using the Stackelberg game-theoretic approach described above.
By leveraging the concept of the Stackelberg game in \cite{bacsar1998dynamic}, we can elucidate the complex interplay between the two agents and identify their best actions. 
Figure~\ref{fig:illustration} visualizes the above-defined game in a networked control system.
\begin{figure} [!t]
    \centering
    \includegraphics[width=0.5\textwidth]{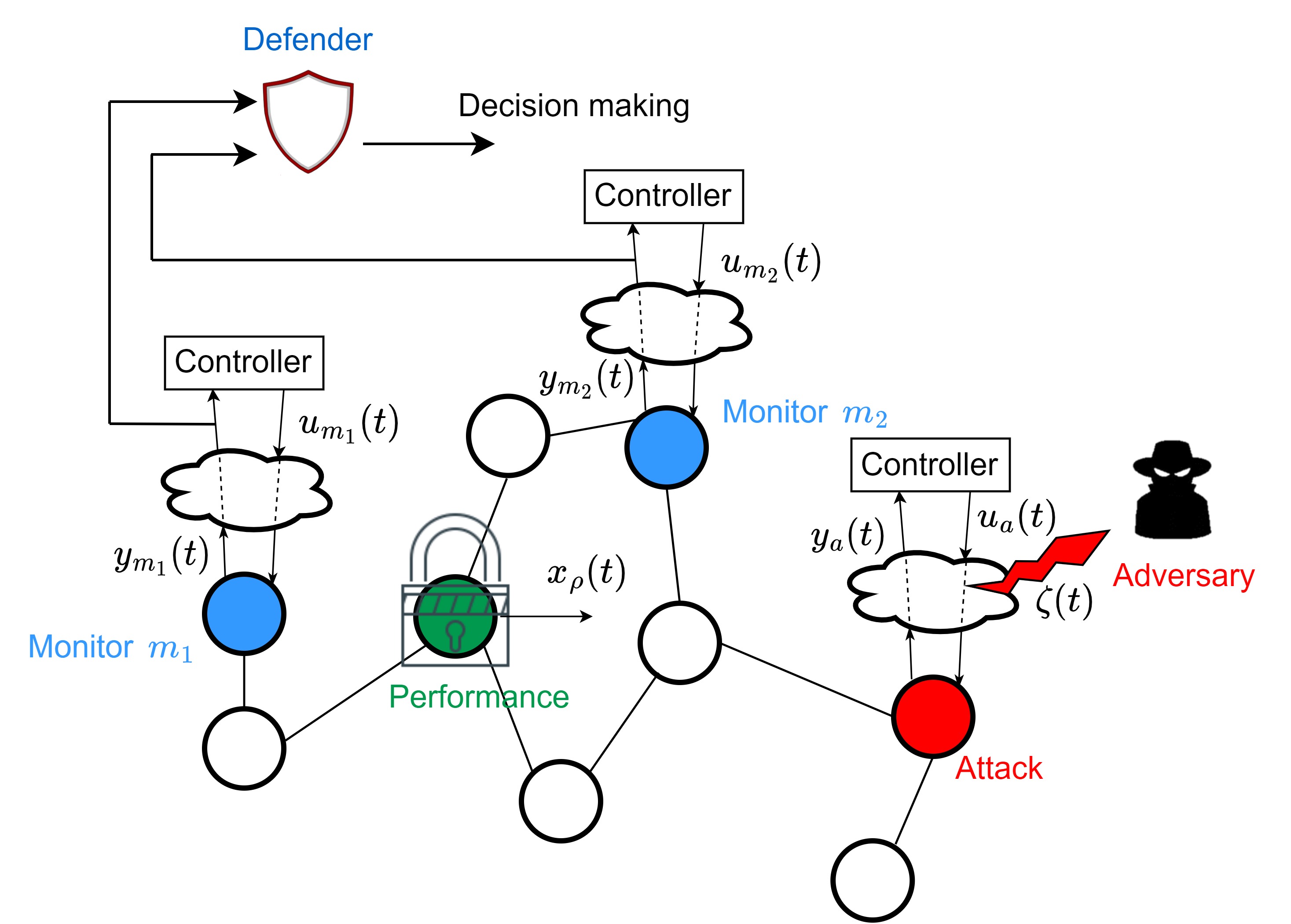}
    \caption{An illustration of a networked control system with the (green) {target} vertex. While the defender places sensors at the (blue) monitor vertices, the adversary conducts a stealthy data injection attack on the (red) attack vertex.}
    \label{fig:illustration}
\end{figure}
The contributions of this paper are the following:
\begin{enumerate}
    \item A novel {defense cost,}
    the expected output-to-output gain, is proposed to capture the expected worst-case impact of stealthy attacks with {an uncertain target vertex}.
    \item The security allocation problem is cast in a Stackelberg game-theoretic framework with the defender as the leader and the malicious adversary as the follower.
    \item We propose a control design 
    that fulfills a graph-theoretic necessary and sufficient condition under which the boundedness of the {defense} cost is guaranteed.
    \item Leveraging the uncertainty of {the target vertex,}
    we show that the necessary and sufficient condition in 3) restricts the admissible choice of monitor sets to be dominating sets of the graph.
    \item The advantage of the proposed security allocation scheme is highlighted through the alleviation of the computational complexity.
\end{enumerate}

The remainder of this paper is organized as follows.
Section \ref{sec:prob_for} describes a networked control system under stealthy attacks {and the adversarial modeling. Then, Section~\ref{sec:strategies} presents how a malicious adversary and a defender design their strategies.}
Thereafter, Section~\ref{sec:feasibility} investigates the boundedness of the {defense} cost and the worst-case impact of stealthy attacks caused by the malicious adversary. The investigation affords us to restrict the admissible actions of the defender, which is presented at the end of Section~\ref{sec:feasibility}. In Section~\ref{sec:game}, by employing the Stackelberg game-theoretic framework, we {propose two solutions to find} the optimal actions for the malicious adversary and the defender. In Section~\ref{sec:comp}, the effectiveness of the proposed security allocation scheme in terms of computational complexity is highlighted, especially in large-scale networks. 
Section~\ref{sec:num} presents a numerical example to validate the obtained results while 
Section~\ref{sec:concl} concludes the paper. We conclude this section by providing the notation to be used throughout this paper.
~\\~\\~
\textbf{Notation:} the set of real positive numbers is denoted as $\Rbb_+$ ; $\Rbb^n$ and $\Rbb^{n \times m}$ stand for sets of real $n$-dimensional vectors and $n$-row $m$-column matrices, respectively.
A vector with the $i$-th element set to $1$ and the rest of the elements to zero is denoted
$e_i \in \Rbb^n$.
For a set $\Ac$, $|\Ac|$ stands for the set cardinality. 
{For a given discrete random variable $y \in D_y$ having a probability mass function $p$, the expected value of a function $f(y)$ is denoted as $\Ebb_{y \sim p} \left[ f(y) \right] = \sum_{y \in D_y} \, p(y) f(y)$}.
A continuous linear time-invariant (LTI) system with the state-space model $\dot{x}(t) = \bar Ax(t) + \bar Bu(t),\; y(t) = \bar Cx(t) + \bar Du(t)$ is denoted as $\bar \Sigma \triangleq (\bar A,\bar B,\bar C,\bar D)$.
The space of square-integrable  functions is defined as $\Lc_{2} \triangleq \bigl\{f: \Rbb_{+} \rightarrow \Rbb ~|~ \norm{f}^2_{\Lc_2 [0,\infty]} < \infty \bigr\}$ and the extended space be defined as $\Lc_{2e} \triangleq \bigl\{ f: \Rbb_{+} \rightarrow \Rbb ~|~ \norm{f}^2_{\Lc_2 [0,T]} < \infty,~ \forall~ 0 < T < \infty \bigr\} $.
The notation $\norm{x}^2_{\Lc_2}$  is used  as shorthand for the norm $\norm{x}_{\Lc_2 [0,T]}^2 \triangleq \frac{1}{T}\int_{0}^{T} \norm{x(t)}_2^2~\text{d}t$ if the time horizon $[0,T]$ is clear from the context.
Let $\Gc \triangleq (\Vc, \Ec, A)$ be {an undirected} graph with the set of $N$ vertices $\Vc = \{1, 2,...,N\}$, the set of edges $\Ec \subseteq \Vc \times \Vc $, and the  adjacency matrix $A = [a_{ij}]$.
For any $(i,j) \in \Ec, ~i\neq j$, the element of the adjacency matrix $a_{ij}$ is positive, and with $(i,j) \notin \Ec$ or $i = j$, $a_{ij} = 0$. 
The degree of vertex $i$ is denoted as 
$\Delta_i =  \sum_{j=1}^{n} a_{ij}$ and the degree matrix of graph $\Gc$ is defined as 
$\Delta = {\bf diag}\big(\Delta_1, \Delta_2,\dots, \Delta_N\big)$, where ${\bf diag}$ stands for a diagonal matrix. 
The Laplacian matrix is defined as $L = [\ell_{ij}] = \Delta - A$.
Further, $\Gc$ is called an undirected connected graph if and only if matrix $A$ is symmetric and the algebraic multiplicity of zero as an eigenvalue of $L$ is one.
The set of all neighbours of vertex $i$ is denoted as $\Nc_i = \{j \in \Vc~|~ (i,j) \in \Ec \}$.
We denote the subset of $\Vc$ excluding a vertex $i$ as $\Vc_{-i} \triangleq \Vc \setminus \{i\}$.
\section{Problem Description}
\label{sec:prob_for}
We first describe a networked control system under stealthy attacks in the presence of a defender and a malicious adversary. 
{Then, the malicious goal and the attack strategy are modeled in the remainder of this section.}
\subsection{Networked control system under stealthy attacks}
Consider an undirected connected graph $\Gc \triangleq (\Vc, \Ec, A)$ with $N$ vertices, the state-space model of a one-dimensional vertex $i$ is described:
	\begin{align}
	\dot x_i(t) &= u_i(t), ~~ i \in \Vc = \bigl\{1,~2,\ldots,~N\bigr\},
	\label{sys:xi1}
	\end{align}
where $x_i(t) \in \Rbb$ {and $u_i(t) \in \Rbb$ are the state and the local control input of vertex $i$, respectively}. 
Each vertex $i \in \Vc$ is controlled by the control law:
	\begin{align}
	u_i(t) = -\theta_i x_i(t) +  \sum_{j \in \Nc_i} \big(x_j(t) - x_i(t)\big),
	\label{sys:u}
	\end{align}
where $\theta_i \in \Rbb_+$ is an adjustable self-loop control gain at vertex $i$.
This self-loop control gain will be used to improve the security of the entire network later in this paper.
For convenience, let us denote $x(t)$ as the state of the entire network, $x(t) = \big[x_1(t),~x_2(t),\ldots,~x_N(t)\big]^\top$.

To get prepared for facing malicious activities,
the defender selects a subset of the vertex set $\Vc$ as a set of monitor vertices, denoted as $\Mc = \{m_1,m_2,\ldots,m_{| \Mc |} \}$, on which to place a sensor at each selected monitor vertex. 
Due to practical reasons, the number of utilized sensors should be constrained.
Let us denote $n_s$ as the sensor budget that is the maximum number of utilized sensors, i.e., $|\Mc| \leq n_s$.

{
Given the defense strategy,} 
the malicious adversary selects a vertex $a \in \Vc$ on which to conduct an additive time-dependent attack signal $\zeta(t) \in \Rbb$, where $\zeta \in \Lc_{2e}$, at its input as follows:
\begin{align}
	u_a(t) = -\theta_a x_a(t) + \sum_{j \in \Nc_a} \big(x_j(t) - x_a(t)\big) + \zeta(t).
	\label{sys:ua}
	\end{align}

{Based on the above descriptions of the network, the defense strategy, and the malicious plan,}
the system model \eqref{sys:xi1} under the control law \eqref{sys:u}-\eqref{sys:ua} can be rewritten in the presence of the attack signal at the {attack vertex $a$ with output of the target vertex $\rho$}
and outputs observed at the monitor vertices $m_k \in \Mc$ as follows:
	\begin{align}
	\dot x(t) &= - \bar L x(t) + e_a \zeta(t),
	\label{sys:x}
	\\
	y_\rho (t) &= e_\rho^\top x(t),
	\label{sys:yt}
	\\
	y_{\Mc} (t) &= C_\Mc^\top x(t),
	\label{sys:ym}
	\end{align}
	where
	$C_\Mc = [e_{m_1},e_{m_2},\ldots,e_{m_{|\Mc|}}]$, $\bar L = L + \Theta$, and $\Theta = \textbf{diag}(\theta_1,\theta_2,\ldots,\theta_N)$. 
The Laplacian matrix $L$ is associated with the undirected connected graph $\Gc$ and $\theta_i \in \Rbb_+,~\forall i \in \Vc$, resulting in that all the eigenvalues of the matrix $\bar L$ are positive real.
This property of $\bar L$ ensures that the state of the network $x(t)$ asymptotically converges to the origin in the attack-free case, affording us to employ the following assumption.
\begin{Assumption}
    The system \eqref{sys:x} is at its equilibrium $x_e = 0$ before being affected by the attack signal $\zeta(t)$.  \QET
\end{Assumption}
{
\begin{Remark}
    The system \eqref{sys:x}-\eqref{sys:ym} is guaranteed to be asymptotically stable in the attack-free case. 
    Unfortunately, the stability of an attack-free system is not enough to determine the 
    impact of stealthy additive false data injection attacks \eqref{sys:ua}, which are mainly studied in this paper. The attack impact needs to be evaluated through the invariant zeros of the system \eqref{sys:x}-\eqref{sys:ym}, which will be described in Section~\ref{sec:feasibility}. 
    \QET
\end{Remark}
}

\subsection{Stealthy attack model}
The purpose of the malicious adversary is to maximally disrupt {a distant target vertex (denoted as $\rho$) by compromising an attack vertex $a$,} while remaining stealthy to the defender {(see the discussion on the importance of the stealthiness in \cite[Sec. II.E]{umsonst2021bayesian}). This attack strategy is motivated by existing scenarios considered in the literature such as a single target vertex in network control systems \cite{pirani2021game}, malicious control in competitive power systems \cite{demarco1996potential}, 
Crossfire attacks in computer security \cite{kang2013crossfire}, and adversarial reachable sets \cite{pirani2021strategic}, where the malicious goal is to impact other vertices beyond the initially compromised vertex. Based on these motivating examples, we employ the following assumption.
\begin{Assumption}
	\label{assumption:attack_not_performance}
     {For any given attack vertex $a$,
	  the target vertex $\rho$} is distinct from the attack vertex $a$, i.e., $\rho \in \Vc_{-a}$.
\end{Assumption}
}

{The above malicious purpose allows us to}
mainly focus on the stealthy data injection attack that will be defined in the following.
Consider the above structure of the continuous LTI system \eqref{sys:x}-\eqref{sys:ym}, which we denote as $\Sigma_{\rho \Mc} \triangleq (-\bar L,e_a,[e_\rho,C_\Mc]^\top,0)$, with 
the monitor outputs $y_{m_k} (t) = e_{m_k}^\top x(t),~\forall m_k \in \Mc$.
The input signal $\zeta(t)$ of the system $\Sigma_{\rho \Mc}$ is called the stealthy data injection attack if the monitor outputs satisfy  $\norm{y_{m_k}}_{\Lc_{2}}^2 < \delta_{m_k}$, 
for all  $m_k \in \Mc$, in which $\delta_{m_k} > 0$ is given for each corresponding monitor vertex $m_k$ and called an alarm threshold.
This means that the adversary is said to be detected if there exists at least one monitor vertex $m_k \in \Mc$ whose output energy crosses its corresponding alarm threshold $\delta_{m_k}$.
The impact of the stealthy data injection attack is measured via the output energy of {the target vertex $\rho$} over the horizon $[0, T]$, i.e., $\norm{y_\rho}_{\Lc_{2}[0, T]}^2$. 
{This performance specification is commonly used in the literature on secure control systems \cite{mo2009secure,ren2021kullback,zhu2015game,wu2020zero,pirani2021game} where it captures the average impact on the target in a certain time interval. 
The long time horizon considered in the attack impact is motivated by a common assumption in the literature that adversaries aim for a long-term malicious impact after they have made a significant investment to infiltrate the network and acquire system parameters \cite{umsonst2021bayesian}.
Further, akin to the $\Hc_-/\Hc_\infty$ metrics \cite{wang2007lmi} and the LQ controller design, designing problems with energy costs and linear systems can be formulated into tractable problems.}

The worst-case impact of the stealthy data injection attack conducted by the malicious adversary on {the target}
will be further investigated. Then, this worst-case attack impact will be utilized to formulate the objectives of the adversary and the defender in the following {section}.
\section{Attack and Defense Strategies}
\label{sec:strategies}
{In the first two parts of this section, the malicious and defense objectives are formulated to design the strategies for the malicious adversary and the defender. In the remainder of this section, the security allocation methodology that is the main focus of this paper is presented to show how the defender designs their defense strategy.}
\subsection{Attack strategy}
{Practically, the malicious adversary designs their attack policies after acquiring enough system parameters and observing defense strategies. Given the set of monitor vertices $\Mc$, the malicious adversary selects an attack vertex $a$ that maximizes the following worst-case impact of stealthy attacks on the distant target vertex $\rho \in \Vc_{-a}$:}
	\begin{align}
	J_\rho(a,\Mc)  \triangleq \sup_{x(0) \, = \, 0,~\zeta \, \in \, \Lc_{2e}} &\norm{y_\rho}_{\Lc_{2}}^2 \label{Jtau1} \\
	\text{s.t.}~~~~~~
	&\norm{y_{m_k}}_{\Lc_{2}}^2 \leq \delta_{m_k},~\forall m_k \in \Mc,
	\non \\ 
    & {\eqref{sys:x}-\eqref{sys:ym}.} \non
	\end{align}
The dual problem of \eqref{Jtau1} is given as follows:
	\begin{align}
	\inf_{\gamma_{m_k} > 0}& ~ \Bigg[  \sup_{x(0) \,= \,0,~\zeta \, \in \, \Lc_{2e}} ~ 
	\Big( \norm{y_\rho}^2_{\Lc_{2}} - \sum_{m_k \in \Mc} \gamma_{m_k}  \norm{y_{m_k}}_{\Lc_{2}}^2    \Big) \non \\ 
	&\hspace{4cm}
	+ \sum_{m_k \in \Mc} \gamma_{m_k} \delta_{m_k} \Bigg]
	\label{dual_form} \\
    { \text{s.t.}}~~&~~~  
    {\eqref{sys:x}-\eqref{sys:ym}.} \non
	\end{align}
	The dual problem \eqref{dual_form} is bounded only if $\norm{y_\rho}^2_{\Lc_{2}} - \sum_{m_k \in \Mc} \gamma_{m_k}  \norm{y_{m_k}}_{\Lc_{2}}^2 \leq 0,~\forall \zeta \in \Lc_{2e}$ and $x(0) = 0$, which results in the following minimization problem:
	\begin{align}
	J_\rho(a,\Mc) = \min_{\gamma_{m_k} > 0}& ~ \sum_{m_k \in \Mc} \gamma_{m_k}  \delta_{m_k}, \label{security_metric} \\
	\text{s.t.}~~& \norm{y_\rho}^2_{\Lc_{2}} - \sum_{m_k \in \Mc} \gamma_{m_k} \norm{y_{m_k}}_{\Lc_{2}}^2 \leq 0, \non \\
	&~~~~
	{\eqref{sys:x}-\eqref{sys:ym}},~ 
    x(0) = 0,~\forall \zeta \in \Lc_{2e}. \non
	\end{align}
The strong duality can be proven by utilizing S-Procedure \cite[Ch. 4]{petersen2000robust}. Recalling the key results in dissipative system theory for linear systems with quadratic supply rates \cite{trentelman1991dissipation}, the optimization problem {\eqref{security_metric} can be translated into the following semidefinite programming (SDP) problem \cite[Prop. 1]{teixeira2021security}:}
\begin{align}
	J_\rho(a,\Mc) =& \min_{\gamma_{m_k} > 0,~P=P^\top \geq 0}~~  \sum_{m_k \in \Mc} \gamma_{m_k}  \delta_{m_k} \label{opt_LM_security_metric} \\
	&\text{s.t.}~ \ba{cc}
	-\bar L P - P  \bar L & P  e_a \\
	e_a^\top P & 0
	\ea  
	+ \ba{c}
	e_\rho \\ 0
	\ea \ba{cc}
	e_\rho^\top & 0
	\ea  \non \\ 
	&  ~~~ 
	- \sum_{m_k \in \Mc} \gamma_{m_k} \ba{c}
	e_{m_k} \\ 0
	\ea 
	\ba{cc}
	e_{m_k}^\top & 0
	\ea \leq 0. \non
	\end{align}

{It is worth noting that} to guarantee the existence of a solution to the optimization problem \eqref{opt_LM_security_metric}, we need to show the {boundedness} of the optimization problem \eqref{Jtau1} \cite{teixeira2015strategic}, which will be discussed in Section~\ref{sec:feasibility}.
{The following subsection presents how the defender designs their defense strategy without knowing the exact target of the malicious adversary.}
\begin{Remark}
    In a similar scenario, another objective function based on $\Lc_2$-gain for both the adversary and the defender has been proposed in \cite[Sec. 3]{pirani2021game}.
	The objective function in \cite[Sec. 3]{pirani2021game} was formulated in terms of the maximal $\Lc_2$-gains from the attack vertex $a$ to the {target} vertex $\rho$ and from the attack vertex $a$ to the monitor vertex $m_k$. More specifically, the objective function in \cite[Sec. 3]{pirani2021game} is given by
	$$W_\rho(a, m_k) = \sup_{ \|\zeta\|_{\Lc_2} \neq 0}
	\frac{\|y_\rho\|_{\Lc_2}^2}{\|\zeta\|_{\Lc_2}^2}
	 - \lambda \sup_{ \|\zeta\|_{\Lc_2} \neq 0}
	\frac{\|y_{m_k}\|_{\Lc_2}^2}{\|\zeta\|_{\Lc_2}^2},~ (\lambda \geq 0). $$
	The above objective $W_\rho(a, m_k)$ also considers two different outputs $y_\rho(t)$ and $y_{m_k}(t)$, but note that the output energies are maximized separately, thus leading to two different optimal input signals $\zeta(t)$ in general cases.
	By contrast, our objective function \eqref{Jtau1} considers the worst-case impact of stealthy attacks that is simultaneously characterized by the multiple outputs $y_\rho(t)$ and $y_{m_k}(t)$ with respect to a single input signal $\zeta(t)$. \QET
\end{Remark}
\subsection{Defense strategy}
{
The malicious adversary aims at maximizing \eqref{Jtau1} with respect to a distant target vertex $\rho$, which is uncertain to the defender. This assumption closely aligns with practical situations where the defender seldom foresees the exact intentions of malicious adversaries. To design a suitable defense strategy despite such uncertainty, the defender can conduct a risk assessment \cite{ross2012guide} to assess and reason about the impact and the likelihood of potential malicious activities (namely pairs of attack and target vertices in our context).
To this end, the defender considers the malicious target, represented by the location of the vertex $\rho$, in a probabilistic manner. 
For a given attack vertex $a$, the uncertain target vertex $\rho$ is characterized probabilistically through a conditional belief $\pi_a(\rho)$, which is assumed to be positive $\forall \rho \in \Vc_{-a}$. This belief model aligns partially with the concept of attack types in games with incomplete information \cite{umsonst2021bayesian,harsanyi1967games}.
Therefore, instead of minimizing \eqref{Jtau1} as in games with complete information, the defender utilizes the above-defined conditional belief $\pi_a(\rho)$ to consider an expected worst-case impact of stealthy attacks as a proxy for \eqref{Jtau1}. Then, the defender desires to choose a set of monitor vertices $\Mc$ that minimizes the following defense cost:
}
\begin{align}
R(a,\Mc) \triangleq \cfr(|\Mc|) + Q(a,\Mc), 
\label{game_payoff_def} 
\end{align}
{
where the expected worst-case impact of stealthy attacks is defined as:
\begin{align}
    Q(a,\Mc) &\triangleq \Ebb_{\rho\sim\pi_a} \left[ J_\rho (a,\Mc) \right] \non \\
    & = \sum_{\rho \in \Vc_{-a} } {\pi_a(\rho)} J_\rho (a,\Mc),
    \label{game_payoff_Q} 
\end{align}
}
and 
$\cfr(|\Mc|)$ is a cost for the number of utilized sensors, {which is assumed to be bounded for any monitor set $\Mc \subseteq \Vc$}. 
{
In the following subsection, we present how the defender plans their defense strategy by addressing the above-defined objectives \eqref{game_payoff_def}-\eqref{game_payoff_Q}.}
\subsection{Security Allocation Methodology}
{
The security allocation entails that the defender strategically allocates defense resources to monitor specific vertices, aiming to enhance the security level of the network. 
The strategic selection of a monitor set is computed offline in the design phase, where the defender simulates and evaluates all the possible attack scenarios in order to seek the best monitor set that minimizes the defense cost \eqref{game_payoff_def}.
}

{
To accomplish the design, the defender first evaluates the defense cost \eqref{game_payoff_def} based on the potential target of the malicious adversary, which is generally uncertain, in a probabilistic way (see more discussions in \cite[Sec. II.E]{umsonst2021bayesian}). 
Secondly, the defender changes the attack scenario to other vertices and repeats the investigation conducted in the previous step for all potential attack vertices. In the end, the defender obtains the result of the enumeration of all the action scenarios, which is in line with other Stackelberg security games found in the literature \cite{shukla2022robust}. Finally, the result from the enumeration enables the defender to find the best monitor set, which will be discussed in Section~\ref{sec:game}.
}

{
In the steps mentioned above, the defender can neglect monitor sets that yield unbounded defense costs, reducing the defender's action space, saving computing resources, and fostering the design procedure. 
From \eqref{Jtau1}, $J_\rho(a,\Mc)$ is non-negative for every pair of attack vertex $a$ and monitor set $\Mc$. 
Thus, the {defense} cost $R(a,\Mc)$ 
and the expected worst-case impact of stealthy attacks $Q (a,\Mc)$ are 
bounded when the worst-case impact of stealthy attacks \eqref{Jtau1} on every {target}  vertex $\rho \in \Vc_{-a}$ is bounded.
In the following section, we will present how the defender finds a set of admissible monitor vertices $\Mc$ that guarantees the boundedness of the worst-case impact of stealthy attacks \eqref{Jtau1} for every attack vertex. 
}
\section{Characterizing the set of monitor vertices}
\label{sec:feasibility}
In this section, we first provide an upper bound of the worst-case impact of stealthy attacks \eqref{Jtau1}. The {boundedness} of this upper bound is guaranteed by a necessary and sufficient condition.
By analyzing this upper bound, we provide a graph-theoretic necessary and sufficient condition under which the cost \eqref{game_payoff_def} and the expected worst-case impact \eqref{game_payoff_Q} are bounded. This condition, then, allows us to limit the admissible actions of the defender. In the remainder of this section, we show how the admissible actions of the defender are characterized. 
\subsection{Evaluating the worst-case impact of stealthy attacks}
	The following lemma states a key property of the worst-case impact of stealthy attacks \eqref{Jtau1}.
	\begin{Lemma} \label{lem:J_upbound}
		Consider the continuous LTI system $\Sigma_\Mc = (-\bar L,e_a,C_\Mc^\top,0)$ with a given {attack vertex $a$, a target vertex $\rho \in \Vc_{-a}$}, and a non-empty monitor vertex set $\Mc$, the worst-case impact \eqref{Jtau1} has an upper bound: 
		\begin{align}
		{J}_\rho(a,\Mc) ~\leq~ \overline{J}_\rho(a,\Mc),
		\label{lem:Jtau1_upbound}
		\end{align}
		where 
		\begin{align}
		\overline{J}_\rho(a,\Mc) &= \min_{m_k \in \Mc} \left\{\begin{aligned} 
		\sup_{x(0)=0,~\zeta \in \Lc_{2e}} &~~ \norm{y_\rho}^2_{\Lc_{2}} \\
		\text{s.t.}~~~~~&~~ 
		\norm{y_{m_k}}^2_{\Lc_{2}} \leq \delta_{m_k}
		\end{aligned}\right\}.
		\label{lem:Jtau1_upbound_def}
		\QET
		\end{align}
	\end{Lemma}
\begin{proof}
    See Appendix~\ref{app:J_upbound}. 
\end{proof} 

	\textit{Lemma~\ref{lem:J_upbound}} enables us to guarantee the boundedness of the worst-case impact of stealthy attacks \eqref{Jtau1} through considering the isolated worst-case impact of stealthy attacks \eqref{lem:Jtau1_upbound_def} at a single monitor vertex $m_k \in \Mc$. 
    Next, at the first stage in the investigation of the boundedness of the worst-case impact of stealthy attacks \eqref{lem:Jtau1_upbound_def}, we adopt a result in \cite{teixeira2015strategic}.
	Inspired by \cite[Th. 2]{teixeira2015strategic}, the {boundedness} of the optimization problem \eqref{lem:Jtau1_upbound_def} is related to the invariant zeros of $\Sigma_\rho \triangleq (-\bar L,e_a,e^\top_\rho,0)$ and $\Sigma_{m_k} \triangleq (-\bar L,e_a,e_{m_k}^\top,0)$, which are defined as follows.
	\begin{Definition}[Invariant zeros]  \label{def:invariant_zero}
		Consider the strictly proper LTI system $\bar \Sigma \triangleq (\bar A,\bar B,\bar C,0)$ where $\bar A,\bar B$, and $\bar C$ are real matrices with appropriate dimensions. A tuple $(\bar \lambda,\bar{x},\bar g) \in \Cbb \times \Cbb^N \times \Cbb$ is a zero dynamics of $\bar \Sigma$ if it satisfies
		\begin{align}
		\ba{cc}
		\bar \lambda I - \bar A ~~~~ & -\bar B \\
		\bar C & 0
		\ea
		\ba{c}
		\bar{x} \\ \bar g
		\ea
		=
		\ba{c}
		0 \\ 0
		\ea,
		~~~ \bar{x} \neq 0.
		\label{def:inv_zero}
		\end{align}
		A finite $\bar \lambda$ is called a finite invariant zero of the system $\bar \Sigma$.
		The strictly proper system $\bar \Sigma$ always has at least one invariant zero at infinity \cite[Ch. 3]{franklin2002feedback}. Further, invariant zeros that have positive real parts are called unstable invariant zeros.
		\QET
	\end{Definition}

More specifically, to guarantee the boundedness of the worst-case impact of stealthy attacks \eqref{lem:Jtau1_upbound_def}, let us state the following lemma.
\begin{Lemma}[{\cite[Th. 2]{teixeira2015strategic}}]
\label{lem:inv_zero}
	Consider the following continuous LTI systems $\Sigma_\rho \triangleq (-\bar L,e_a, e^\top_\rho,0)$ and $\Sigma_{m_k} \triangleq (-\bar L,e_a,e_{m_k}^\top,{0}),~\forall m_k \in \Mc$.  The optimization problem \eqref{lem:Jtau1_upbound_def} is {bounded} if, and only if, there exists at least one system $\Sigma_{m_k}$ such that its unstable invariant zeros are also invariant zeros of $\Sigma_\rho$.
    \QET
\end{Lemma}
\begin{proof} 
	Follows directly the result in \cite[Th. 2]{teixeira2015strategic}.
\end{proof}

The result in \textit{Lemma~\ref{lem:inv_zero}} prompts us to investigate invariant zeros of $\Sigma_{m_k}$.
Let us adopt the following lemma from our previous work \cite{nguyen2022zero} that considers finite invariant zeros of $\Sigma_{m_k}$.
	\begin{Lemma}[{\cite[Lem. 4.4]{nguyen2022zero}}]
		\label{lem:no_un_zero}
		Consider a networked control system associated with an undirected connected graph $\Gc \triangleq (\Vc,\Ec, A)$, whose closed-loop dynamics is described in \eqref{sys:x}. 
		Suppose that the networked control system is driven by the stealthy data injection attack at a single attack vertex $a$, and observed by a single monitor vertex ${m_k}$, resulting in the state-space model $\Sigma_{m_k} \triangleq (-\bar L,e_a,e^\top_{m_k},0)$. 
		Then, there exist self-loop control gains $\theta_i, ~ \forall i \in \{1,2,\ldots,N\},$ in \eqref{sys:u} such that $\Sigma_{m_k}$ has no 
		finite unstable invariant zero.
		\QET
	\end{Lemma}
\begin{proof}
    See Appendix~\ref{app:lem_pf_no_un_zero}.
\end{proof}

\textit{Lemma~\ref{lem:no_un_zero}} enables us to carefully design the control law \eqref{sys:u}, i.e. select $\theta_i$, such that for every pair of an input vertex $a$ and an output vertex $m_k$, the corresponding LTI system $\Sigma_{m_k} = (-\bar L,e_a,e_{m_k}^\top,0)$ has no unstable invariant zero. 
Hence, it leaves us to investigate infinite invariant zeros of systems $\Sigma_{m_k},~\forall m_k \in \Mc$ in the following subsection.
\subsection{Infinite invariant zeros}
\label{sec:dom_set}
	We investigate the infinite invariant zeros of the systems $\Sigma_{\rho}$ and $\Sigma_{m_k},~\forall m_k \in \Mc$. In the investigation, we make use of known results connecting infinite invariant zeros mentioned in \textit{Definition \ref{def:invariant_zero}} and the relative degree of a linear system, which is defined below.
	\begin{Definition}[Relative degree {\cite[Ch. 13]{khalil2002nonlinear}}] \label{def:rela_deg}
		Consider the strictly proper LTI system $\bar \Sigma \triangleq (\bar A,\bar B,\bar C,0)$ with $\bar A \in \Rbb^{n \times n}$, $\bar B$, and $\bar C$ are real matrices with appropriate dimensions.
		The system $\bar \Sigma$ is said to have relative degree $r ~ (1 \leq r \leq n) $ if the following conditions satisfy
		\begin{align}
		&\bar C \bar A^{k} \bar B = 0, ~~ 0 \leq k < r-1,
		\non  \\
		&\bar C \bar A^{r-1} \bar B \neq 0. 
		\label{def_red}
		\end{align}
		\QET
	\end{Definition}
\begin{Remark} \label{rem:re_deg_tf}
Let $\bar H(s) = \bar C(sI-\bar A)^{-1}\bar B$ be the transfer function of the above system $\bar \Sigma$.
The relative degree $r$ of the system $\bar \Sigma$ defined in \textit{Definition \ref{def:rela_deg}} is also the difference between the degrees of the denominator and the numerator of $\bar H(s)$ \cite{khalil2002nonlinear}, which in turn corresponds to the degree of the infinite zero if $\bar \Sigma$ is minimal realization~\cite[Ch. 3]{franklin2002feedback}.  
\QET
\end{Remark}

Based on \textit{Definition~\ref{def:rela_deg}}, let us denote $r_{(\rho,a)}$ and $r_{(m_k,a)}$ as the relative degrees of $\Sigma_\rho$ and $\Sigma_{m_k},~\forall m_k \in \Mc$, respectively.
	In the scope of this study, we have assumed that the attack signal $\zeta(t)$ in \eqref{sys:ua} has no direct impact on the outputs \eqref{sys:yt} and \eqref{sys:ym}, resulting in strictly proper systems
	$\Sigma_\rho$ and $\Sigma_{m_k}$.
	This implies that the relative degrees $r_{(\rho,a)}$ and $r_{(m_k,a)}$ of the systems $\Sigma_\rho$ and $\Sigma_{m_k}$ are positive, yielding their infinite invariant zeros.
	Let us state the following theorem that considers infinite invariant zeros of the systems $\Sigma_\rho$ and $\Sigma_{m_k}$ to provide a necessary and sufficient condition under which the boundedness of the worst-case impact of stealthy attacks \eqref{lem:Jtau1_upbound_def} is guaranteed.
	\begin{Theorem}  \label{th:red_con} 
		Consider the strictly proper LTI systems $\Sigma_\rho \triangleq (-\bar L,e_a,e^\top_\rho,0)$ and $\Sigma_{m_k} \triangleq (-\bar L,e_a,e^\top_{m_k},0), ~\forall m_k \in \Mc$, in which the systems have the same stealthy data injection attack input \eqref{sys:ua} at a single attack vertex $a \in \Vc_{-\rho}$ but different output vertices \eqref{sys:yt}-\eqref{sys:ym}, i.e., $\rho$ for $\Sigma_\rho$ and $m_k$ for $\Sigma_{m_k}$.
		Suppose the systems $\Sigma_\rho$ and $\Sigma_{m_k}$ have relative degrees $r_{(\rho,a)}$ and $r_{(m_k,a)}$, respectively.
		Then, the worst-case impact of stealthy attacks \eqref{lem:Jtau1_upbound_def} is bounded
		if, and only if, there exists at least one system $\Sigma_{m_k}$ such that the following condition holds
		\begin{align}
		r_{(m_k,a)} \leq r_{(\rho,a)}. \label{cond_red}
		\end{align}
		\QET
	\end{Theorem}
\begin{proof}
    See Appendix~\ref{app:th_pf_red_con}.
\end{proof}

Given an arbitrary attack vertex $a$ and a distant target vertex $\rho \in \Vc_{-a}$, \textit{Theorem~\ref{th:red_con}} hints a solution to {monitoring} malicious activities. The defender chooses a non-empty monitor set $\Mc \subset \Vc$ such that there exists at least one monitor vertex $m_k \in \Mc$ that fulfills the condition \eqref{cond_red}.
The following subsection presents how to find such a monitor set $\Mc$.
{
\begin{Remark}
\label{rem:red_mimo}
    Let us consider the following continuous LTI system $\Sigma_\Mc = (-\bar L,e_a,C_\Mc^\top,0)$ where its input is at the vertex $a$ and its outputs are at monitor vertices $m_k \in \Mc$. By employing the definition of the relative degree of single-input-multiple-output systems, adapted from
    \cite{mueller2009normal}, the relative degree of the system $\Sigma_\Mc$ is the least relative degree from its input to its single monitor vertex. Thus, we need to find at least one monitor vertex $m_k$ such that it fulfills the condition \eqref{cond_red}, resulting in the boundedness of \eqref{lem:Jtau1_upbound_def}. This result eventually allows us to guarantee that the worst-case impact of stealthy attacks in \eqref{Jtau1} is bounded according to the property in \eqref{lem:Jtau1_upbound}.
    \QET
\end{Remark}
} 
\subsection{Admissible monitor sets and dominating sets} 
	Consider a subset $\Mc \subset \Vc$ where its cardinality is not greater than the sensor budget $n_s$, the maximum number of available sensors, i.e., $\Mc = \{ m_1,m_2,\ldots,m_{|\Mc|}\}$ and 
	$|\Mc| \leq n_s$.
    Following the discussions in the previous subsection, a monitor set $\Mc$ is admissible if it contains at least one monitor vertex $m_k \in \Mc$ such that this vertex $m_k$ fulfills the necessary and sufficient condition \eqref{cond_red} in \textit{Theorem~\ref{th:red_con}}.
	This set $\Mc$ is called a dominating set which is defined below.
	\begin{Definition}[Dominating set] \label{def:dominating_set}
		Given an undirected graph $\Gc \triangleq (\Vc, \Ec,A)$, a subset of the vertex set $\Dc \subset \Vc$ is called a dominating set if, for every vertex $u \in \Vc \setminus \Dc$, there is a vertex $ v\in \Dc$ such that $(u,v) \in \Ec$.
		\QET
	\end{Definition}

The following lemma presents a necessary and sufficient condition {that allows us to examine whether a subset of the vertex set $\Vc$ is a dominating set.}
	\begin{Lemma}
		\label{lem:dominating_set}
		Consider an undirected graph $\Gc \triangleq (\Vc,\Ec,A)$,
		a subset $\Mc \subset \Vc$ is a dominating set of $\Vc$ {if, and only if,} the following condition holds 
		\begin{align}
        {e_i^\top \Cc(\Mc) > 0,~ \forall i \in \Vc,}
		\label{dom_set_cond}
		\end{align}
		where $\Cc(\Mc) = (A + I) \sum_{m_k \in \Mc}  e_{m_k}$.
		\QET
	\end{Lemma}
 \begin{proof}
     See Appendix~\ref{app:lem_pf_dominating_set}.
 \end{proof}

	By investigating all the subsets of $\Vc$, we can find all the dominating sets which fulfill the condition \eqref{dom_set_cond}.  Let us make use of the following assumption.
\begin{Assumption}
	\label{assumption:dom_set_exist}
	The vertex set $\Vc$ has at least one dominating set {with a cardinality of at most $n_s$.}
    \QET
\end{Assumption}

Based on  \textit{Assumptions~\ref{assumption:attack_not_performance}-\ref{assumption:dom_set_exist}} and the above results in \textit{Lemma~\ref{lem:J_upbound}} and \textit{Theorem \ref{th:red_con}}, we are now ready to state the following theorem that provides a graph-theoretic necessary and sufficient condition under which the cost \eqref{game_payoff_def} and the expected worst-case impact of stealthy attacks \eqref{game_payoff_Q}, caused by the stealthy data injection attack at an arbitrary attack vertex $a$, are bounded.
\begin{Theorem}
    \label{th:dom_set_bound_impact} 
    Suppose that \textit{Assumptions~\ref{assumption:attack_not_performance}-\ref{assumption:dom_set_exist}} hold.
    Consider the networked control system \eqref{sys:x} associated with an undirected connected graph $\Gc$ where the system has the stealthy data injection attack \eqref{sys:ua} at the input of an arbitrary attack vertex $a$ and outputs \eqref{sys:ym} at monitor vertices $m_k \in \Mc$. 
    The {defense} cost $R(a,\Mc)$ in \eqref{game_payoff_def} and
    the expected worst-case impact of stealthy attacks $Q(a,\Mc)$ in \eqref{game_payoff_Q} are bounded if, and only if, the monitor set $\Mc$ is a dominating set of $\Gc$.
    \QET
\end{Theorem}
\begin{proof} 
    Let us consider the following continuous LTI systems $\Sigma_\rho \triangleq (-\bar L,e_a,e_\rho^\top,0)$ and $\Sigma_{m_k} \triangleq (-\bar L,e_a,e_{m_k}^\top,0),~\forall m_k \in \Mc$.
    The systems have the same stealthy data injection attack at the input of an arbitrary attack vertex $a$ but $\Sigma_\rho$ has an output at an arbitrary {target} vertex $\rho \in \Vc_{-a}$ and $\Sigma_{m_k}$ has an output at monitor vertex $m_k \in \Mc$. 
    Based on \textit{Definition~\ref{def:rela_deg}}, 	\textit{Assumption~\ref{assumption:attack_not_performance}} guarantees that the relative degree of $\Sigma_\rho$ is not lower than one, i.e., $r_{(\rho,a)} \geq 1$.

    We begin by providing sufficiency. \textit{Assumption~\ref{assumption:dom_set_exist}} ensures that there exists at least one dominating set that has at most $n_s$ elements. 
    Therefore, the defender selects the monitor set $\Mc$ as one of such dominating sets. 
    According to \textit{Definitions~\ref{def:rela_deg}-\ref{def:dominating_set}}, 
    there exists at least one system $\Sigma_{m_k}$, where its input is at an arbitrary attack vertex $a$ and its output is at the monitor vertex $m_k~(m_k \in \Mc)$, such that its relative degree is not greater than one, i.e., $r_{(m_k,a)} \leq 1$.
    Based on the above observation, one has 
    $r_{(m_k,a)} \leq 1 \leq r_{(\rho,a)},$
    fulfilling \eqref{cond_red}. 
    From the results in \textit{Theorem~\ref{th:red_con}} and \textit{Lemma~\ref{lem:J_upbound}}, the satisfaction of \eqref{cond_red} allows us to guarantee the boundedness of the worst-case impact of stealthy attacks \eqref{Jtau1}.
    Therefore, the {defense} cost $R(a,\Mc)$ and the expected worst-case impact of stealthy attacks $Q(a,\Mc)$ are bounded based on their definitions in \eqref{game_payoff_def}-\eqref{game_payoff_Q}.

    For necessity, let us present a contradiction argument by assuming that the {defense} cost $R(a,\Mc)$ and the expected worst-case impact of stealthy attacks $Q(a,\Mc)$ are bounded while the monitor set $\Mc$ is not a dominating set of $\Gc$. Based on the definitions of $Q(a,\Mc)$ and $R(a,\Mc)$ in \eqref{game_payoff_def}-\eqref{game_payoff_Q}, they are bounded if, and only if, $J_\rho(a,\Mc)$ is bounded for all pairs of $\rho$ and $a$. Since the attack vertex $a$ can be chosen arbitrarily and the monitor set $\Mc$ is not a dominating set, the attack vertex $a$ can be chosen such that it does not belong to $\Mc$ and none of its neighbors belongs to $\Mc$, resulting in $r_{(m_k,a)} > 1~\forall m_k \in \Mc$. On the other hand, the adversary considers all the possibilities of the target vertex $\rho$ including $(\rho,a) \in \Ec$, resulting in $r_{(\rho,a)} = 1$. The above observation gives us $r_{(\rho,a)} = 1 < r_{(m_k,a)},~\forall m_k \in \Mc$,
    violating the necessary and sufficient condition \eqref{cond_red}. Hence, for this particular pair of $\rho$ and $a$, the worst-case impact of stealthy attacks $J_\rho(a,\Mc)$ is unbounded, contradicting the assumption.
\end{proof}

\textit{Lemma~\ref{lem:dominating_set}} enables us to determine whether a subset of $\Vc$ is a dominating set. 
On the other hand, \textit{Theorem~\ref{th:dom_set_bound_impact}} affords us to restrict the admissible actions of the defender to dominating sets of $\Vc$. 
This step is beneficial to the defender in selecting monitor vertices such that the {defense} cost \eqref{game_payoff_def} and the {expected} worst-case impact of stealthy attacks \eqref{game_payoff_Q} are always bounded.
More detail on how the defender and the malicious adversary select their actions is given in the following section.
{
\begin{Remark}
    \label{rem:parallel}
    The condition \eqref{dom_set_cond} enables us to seek all the dominating sets of a given network. 
    Leveraging the structure of \eqref{dom_set_cond}, we sequentially multiply each row of $(A+I)$ with $\sum_{m_k \in \Mc} e_{m_k}$. Whenever the result is equal to zero, we stop the examination and determine that the subset $\Mc$ is not a dominating set. Otherwise, it is a dominating set. 
    Moreover, since the examination of each subset of the vertex set $\Vc$ is independent, it can be executed by parallel computations with the help of computer clusters.
    \QET
\end{Remark}
}
\section{Stackelberg security game}
\label{sec:game}
In this section, to assist the defender and the malicious adversary in selecting their best actions, we employ the Stackelberg game-theoretic framework where the defender acts as \textit{the leader} and the malicious adversary acts as \textit{the follower} of the game. Subsequently, we provide an algorithm to illustrate the procedure of how the two agents seek their best actions.
\subsection{Game setup}
To investigate the best actions of the defender and the adversary, we assume that they are two strategic players in a game. 
The defender can select at most $n_s$ monitor vertices on which to place one sensor at each selected vertex with the purpose of {monitoring} malicious activities.
Given \textit{Assumption~\ref{assumption:dom_set_exist}}, let us denote the {collection} of dominating sets as $\Dbb$, where each dominating set has at most $n_s$ elements, i.e., $\Dbb= \{\Mc~|~ \Mc \subset \Vc,~|\Mc|\leq n_s,~ \Mc ~\text{satisfies \eqref{dom_set_cond}} \}$.
This {collection} $\Dbb$ is chosen as the action space of the defender.
Meanwhile, the malicious adversary is able to select any vertex to conduct the stealthy data injection attack, i.e., the action space of the malicious adversary is $\Abb = \Vc$.

Based on the catastrophic consequences caused by famous malware such as Stuxnet and Industroyer \cite{falliere2011w32,kshetri2017hacking}, the defender should decide their defense strategy regardless of the presence of malicious adversaries since the defender does not know when adversaries appear. 
Consequently, it is reasonable to let the defender select and announce their action publicly before the presence of the adversary \cite{li2018false,shukla2022robust,yuan2019stackelberg}.
The defender is called \textit{the leader} {while the malicious adversary is called \textit{the follower}}.
The purpose of the defender is to minimize the {defense} cost $R(a,\Mc)$ in \eqref{game_payoff_def} {with knowing that the malicious adversary bases their action on the leader's decision.}
Thus, the leader considers the following problem.
\begin{Problem}
\label{prob:min_game_payoff}
{Given that the malicious adversary maximizes \eqref{game_payoff_Q},}
the defender selects an optimal dominating set 
$\Mc^\star \in \Dbb$ that minimizes the cost \eqref{game_payoff_def}. \QET
\end{Problem}
\begin{table}[!t]
\caption{Components of the Stackelberg security game between a defender and a malicious adversary.}
    \label{tab:intro_security_game}
    \centering
    \begin{tabular}{|l|l|}
        \hline  
       \textbf{Component}  & \textbf{Description}  \\ \hline \hline 
       \textit{Players}  & Defender and Adversary \\ \hline \textit{Model knowledge} & The vertex set $\Vc$, the edge set $\Ec$, the self-loop 
       \\ 
       \textit{of two players} & gains $\theta_i$, the alarm threshold $\delta_i~(\forall i \in \Vc)$ \\ \hline
       \textit{Action Space} & Defender: $\Dbb= \{\Mc~|~ \Mc \subset \Vc,|\Mc|\leq n_s, \eqref{dom_set_cond} \}$ \\
       & Adversary: $\Abb = \Vc$ \\ \hline 
       \textit{Game Payoff} & Defender minimizes $R(a,\Mc)$ defined in \eqref{game_payoff_def} \\ 
       \textit{\& Goal}& Adversary maximizes $Q(a,\Mc)$ defined in \eqref{game_payoff_Q} \\ \hline 
       \textit{Information} & Defender takes action first \\
       \textit{Structure} & Adversary responds to Defender's action \\ \hline
    \end{tabular}
\end{table}

We cast the \textit{Problem~\ref{prob:min_game_payoff}} in the Stackelberg game-theoretic framework with the defender as the leader, who selects and announces their action first, and the malicious adversary as the follower. 
The components of the Stackelberg game between the defender and the malicious adversary are summarized in Table~\ref{tab:intro_security_game}. 
This Stackelberg game always admits an optimal action \cite{bacsar1998dynamic}, which is defined below.
\begin{Definition}[Stackelberg optimal action {\cite{simaan1973stackelberg}}] 
    \label{def:Stac_opt}
    If there exists a mapping $\Tc: \Dbb \ra \Abb$ such that, for any fixed $\Mc \in \Dbb$, one has $Q(\Tc \Mc,\Mc) \geq Q(a,\Mc)$ for all $a \in \Abb$, and if there exists  $\Mc^\star \in \Dbb$ such that $R(\Tc \Mc^\star,\Mc^\star) \leq R(\Tc \Mc,\Mc)$ for all $\Mc \in \Dbb$, then the pair $(a^\star(\Mc^\star),\Mc^\star) \in \Abb \times \Dbb$, where $a^\star(\Mc^\star) = \Tc \Mc^\star$, is called a Stackelberg optimal action with the defender as the leader and the adversary as the follower of the game.
    \QET
\end{Definition}

Based on \textit{Definition~\ref{def:Stac_opt}}, we first analyze the Stackelberg optimal action and then provide two solutions that find it in the following subsection.
\subsection{Stackelberg optimal action}
Recall \textit{Problem~\ref{prob:min_game_payoff}} and \textit{Definition~\ref{def:Stac_opt}}, the defender finds their optimal action by solving the following optimization problem:
\begin{align}
    \Mc^\star &= \arg \min_{\Mc \, \in \, \Dbb} ~ R (a^\star(\Mc),\Mc),\label{Eq_M_star} 
\end{align}
where
\begin{align}
    a^\star(\Mc) &= \arg \max_{a \, \in \, \Abb} ~{Q} (a,\Mc). \label{Eq_a_star}
\end{align}
One can verify that
the optimal solution $(a^\star(\Mc^\star),\Mc^\star)$ found through solving the optimization problems~\eqref{Eq_M_star}-\eqref{Eq_a_star} is equivalent to the one in \textit{Definition~\ref{def:Stac_opt}}.
{
To obtain the best monitor set $\Mc^\star$ by solving \eqref{Eq_M_star}, we provide the following proposition.
\begin{Proposition}
    \label{proposition:bigLMI}
    Consider the networked control system \eqref{sys:x}-\eqref{sys:ym} associated with an undirected connected graph $\Gc$.  Denote $\Dbb$ as a non-empty collection of all the dominating sets $\Dc_i$ of the graph $\Gc$ with a cardinality of at most $n_s$, i.e., $\Dbb = \{ \Dc_1, \Dc_2, \ldots, \Dc_{|\Dbb|} \}$.  For each dominating set $\Dc_i$, let $z_i = \sum_{m \, \in \, \Dc_i} e_{m}$ be an $N$-dimensional binary vector, where $j$-th entry of $z_i$ being equal to $1$ indicates that vertex $j$ belongs to $\Dc_i$.  Define $v$ as a $|\Dbb|$-dimensional binary vector where $i$-th entry of $v$ being equal to $1$ indicates that $D_i$ is chosen as a monitor set. Then, the optimal monitor set is determined by $v^\star$, which is the optimal solution to the following mixed-integer semidefinite programming problem:
    \begin{align}
        &  \hspace{-0.8cm}
        \min_{\bar z_{(a,\rho)} \, \in \, \Rbb^N, \, v \, \in \, \{0,1\}^{|\Dbb|}, \, P_{(a,\rho)} \, \in \, \Rbb^{N \times N}, \, \beta \, > \, 0}~~ \cfr(|\Mc|) +  \beta \label{mixedintegersdp} \\
        ~~~~~&\text{s.t.}~~
        \textbf{1}_{|\Dbb|}^\top v = 1,~
        P_{(a,\rho)} = P_{(a,\rho)}^\top \geq 0, \non \\
        &~~~~
        \sum_{\rho \in \Vc_{-a}} \pi_a(\rho) \, \delta^\top \bar z_{(a,\rho)} \leq \beta,
        \non \\
        &~~~~
        0 \leq \bar z_{(a,\rho)} \leq \tilde M \, [z_1,z_2,\ldots,z_{|\Dbb|}] \, v, \non \\
        &
        \ba{cc}
        - \bar L P_{(a,\rho)} - P_{(a,\rho)} \bar L & P_{(a,\rho)} e_a \\
        e_a^\top P_{(a,\rho)} & 0 
        \ea + \textbf{diag} \left( \ba{c} e_\rho \\ 0 \ea \right)  \non \\
        &\hspace{0.5cm}
        - \textbf{diag} \left( \ba{c} \bar z_{(a,\rho)} \\ 0 \ea \right) \leq 0,~~
        \forall \, (a \, , \,\rho) \in \Vc \times \Vc_{-a}, \non 
    \end{align}
    where $\textbf{1}_{|\Dbb|}$ stands for a ${|\Dbb|}$-dimensional all-one vector,
    $\delta = [\delta_1,\delta_2,\ldots,\delta_N]^\top$ is the alarm threshold vector of all the vertices, and
    $\tilde M$ is a large positive number, also called a ``big M'' \cite{milovsevic2023strategic}. \QET
\end{Proposition}
\begin{proof}
    See Appendix~\ref{app:propositionbigLMI}.
\end{proof}
}

{
Successfully solving \eqref{mixedintegersdp} gives us the best monitor set $\Mc^\star$ represented by the optimal solution $v^\star$. However, dealing with the large mixed-integer SDP \eqref{mixedintegersdp}, which contains $N \times (N-1)$ attack scenarios for all possible pairs $(a,\rho)$, poses efficiency challenges in very large networks. To deal with such an issue, we leverage parallel computations mentioned in \textit{Remark~\ref{rem:parallel}} and propose \textit{Algorithm~\ref{al:parallel}.} 
The following section discusses how the proposed concept of dominating sets significantly alleviates the computational complexity in large networks.
}

\begin{algorithm}[!t]
    \caption{Stackelberg optimal action {through parallel computations} \label{al:parallel}}
    \begin{algorithmic}[1]
    \Statex{{\bf Input:} { The vertex set $\Vc$, the edge set $\Ec$, the self-loop gains $\theta_i$, the alarm thresholds $\delta_i,~\forall i \in \Vc$, the sensor budget $n_s$, the cost of utilized sensors $\cfr(|\Mc|)$, the {conditional belief $\pi_a(\rho)$, and $n_c$ computer cores where $j$-th core is denoted as $U_j,~ j \in \{1,2,\ldots,n_c\}$.}
    }
   }		
   \Statex{{\bf Output:}} {
   The best monitor set $\Mc^\star$ and the best attack vertex $a^\star(\Mc^\star)$.
   }
   \Statex{{\bf Initialize:}} { {$\Dbb= \{\Mc \, | \, \Mc \subset \Vc,~|\Mc|\leq n_s,~ \text{\eqref{dom_set_cond}} \}$} }
   \State {{Equally divide $\Dbb$ into $n_c$ partitions $\{ \Dbb_1,\Dbb_2,\ldots,\Dbb_{n_c} \}$ where partition $\Dbb_j$ is assigned to computer core $U_j$.}}
   \For {{every computer core $U_j$}}
   \For {every dominating set $\Mc$ in $\Dbb_j$}
   \For {every pair $(a \, , \, \rho) \in \Vc \times \Vc_{-a}$}
   \State {solve \eqref{opt_LM_security_metric}.}
   \EndFor
   \State {Compute \eqref{game_payoff_def}.}
   \EndFor
   \EndFor
   \State Solve \eqref{Eq_M_star} to obtain $\Mc^\star$ and $a^\star(\Mc^\star)$.
   \end{algorithmic}
\end{algorithm}

\section{Computational Complexity}
\label{sec:comp}
In this section, we highlight the benefits of characterizing admissible monitor sets as dominating sets to the computation, especially in large-scale networked control systems. 

{
Without the consideration of the collection $\Dbb$ in \textit{Proposition~\ref{proposition:bigLMI}}, the security allocation problem requires the defender to solve \eqref{mixedintegersdp} with the collection of all the subsets of the vertex set $\Vc$, which is $\Sbb$ defined in \eqref{app:ZSbb} (see Appendix~\ref{app:propositionbigLMI}). Next, we compute the cardinality of $\Sbb$.
Let us denote the cardinality of $\Sbb$
}
as $S(N,n_s)$ where $N$ is the number of vertices in the network and $n_s$ is the sensor budget. This number $S(N,n_s)$ can be computed as follows:
\begin{align}
    S(N,n_s) = \sum_{k = 1}^{n_s} {{N}\choose{k}}. \label{no_subset}
\end{align} 
This number $S(N,n_s)$ grows dramatically when either the number of vertices $N$ or the sensor budget $n_s$ increases due to $S(N,n_s) = \Oc(N^{n_s})$, where $\Oc$ stands for Big O notation. 

\begin{figure}[!t]
    \centering
    \includegraphics[width=0.5\textwidth]{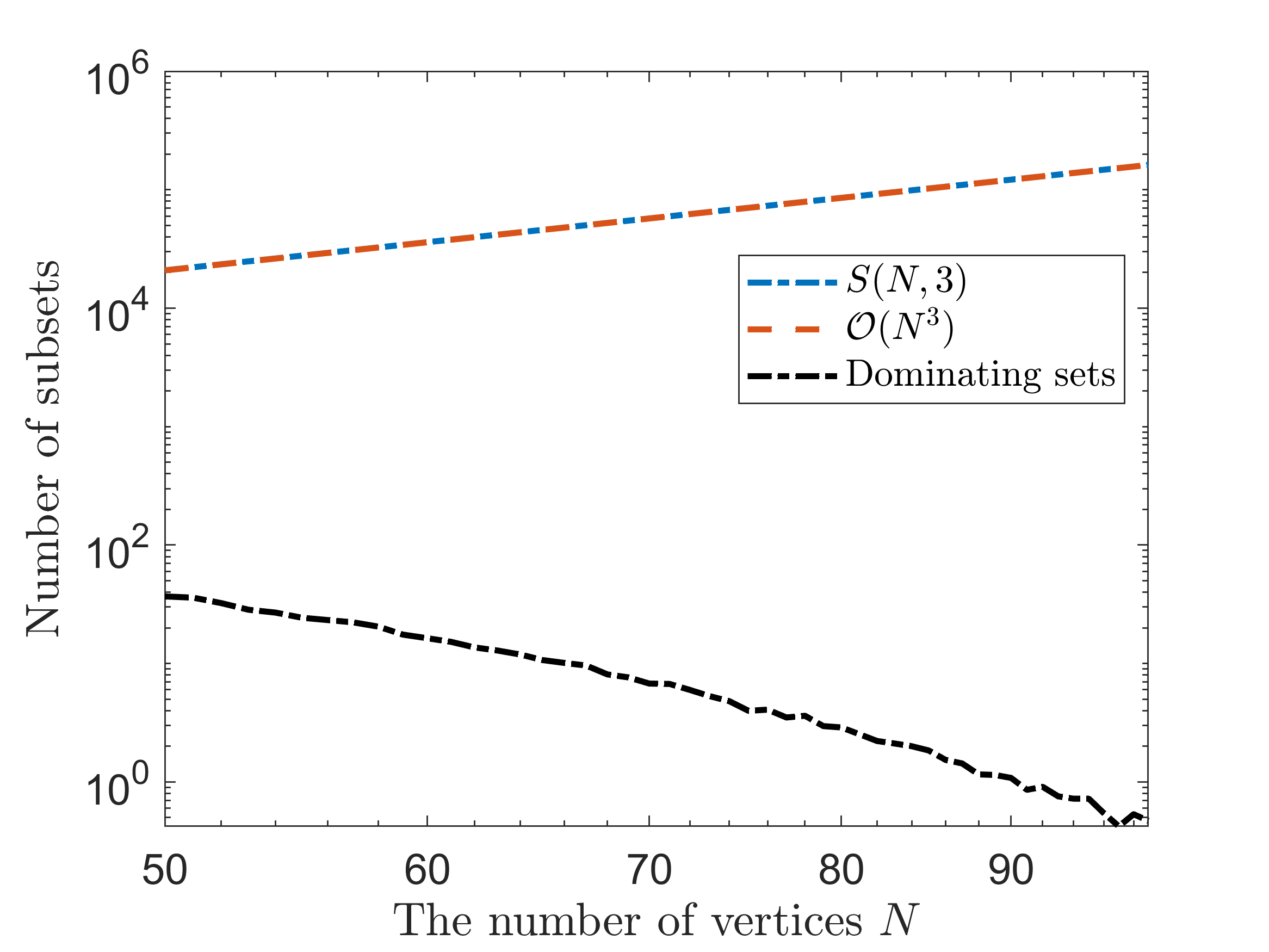}
    \caption{Given $n_s = 3$, the number of subsets of the vertex set $\Vc$ with respect to the number of vertices has the same slope as $\Oc(N^3)$. The average number of dominating sets is given through the Monte-Carlo simulation with 500 samples.}
    \label{fig:no_subset}
\end{figure}

An illustration of the dramatic increase of $S(N,n_s)$ with respect to $N$ (blue dashed-dotted line) can be found in Figure~\ref{fig:no_subset} where it has the same slope as $\Oc(N^3)$ (red dashed line).
In Figure~\ref{fig:no_subset}, we also conduct Monte-Carlo simulations with 500 samples to count the number of dominating sets with respect to the size of the graph $N$, which is denoted as the black dashed-dotted line. 
In the Monte-Carlo simulations, we examine Erdős–Rényi random undirected connected graphs $G(N,q)$ where $N$ is the number of vertices and an edge is included to connect two vertices with probability $q = 0.5$ \cite{bollobas1998random}.
{By observing the results in Figure~\ref{fig:no_subset}, the number of dominating sets with a fixed sensor budget dramatically decreases when the size of networks increases. 
Therefore, adopting the concept of dominating sets in solving \eqref{mixedintegersdp} enables us to significantly reduce its computational complexity.
}

{Regarding the parallel computing proposed in \textit{Algorithm~\ref{al:parallel}}, the concept of dominating sets plays a crucial role in practice.
The number of possible actions computed by \eqref{no_subset} possibly requires the defender to employ a large number of working hours of computer cores, which are limited in practice. In fact,
}
as the number of vertices increases, the number of possible actions grows significantly (see Figure~\ref{fig:no_subset}), making {the solution methodology impractical in dealing with the security problem in very large networks.}
In contrast, the number of dominating sets {with a fixed sensor budget} typically decreases with respect to the size of random graphs (as seen in the example in Figure~\ref{fig:no_subset}), {requiring greatly fewer working hours of computer cores.
As a result, the concept of dominating sets enables us to practically handle the security allocation problem in very large networks.
}

{
It is worth noting that the proposed method is carried out in the design phase, which can be computed offline. Moreover, the proposed method can be improved by utilizing parallel computing toolboxes and computer clusters as discussed above (see \textit{Remark~\ref{rem:parallel}} and \textit{Algorithm~\ref{al:parallel}}). 
Thus, the computational complexity does not significantly impact the implementation of the proposed method (see more discussions in \cite[Sec. V]{shukla2022robust}).}
In the next section, we are likely to show the effectiveness of the proposed security allocation scheme with the notion of dominating sets through a numerical example.
\section{Numerical example}
\label{sec:num}
In the first part of this section, 
{
we validate the main result of this paper presented in \textit{Theorem~\ref{th:dom_set_bound_impact}} and find the Stackelberg optimal action for the defender and the malicious adversary in a numerical example.
}
In the remainder of this section, the alleviation in the computational complexity will be discussed.
{The simulation is performed using Matlab 2023b with YALMIP 2021 \cite{lofberg2004yalmip} and MOSEK solver on a personal computer with 2.9-GHz, 8-core Intel i7-10700 processor and  16 GB of RAM.}

To demonstrate the obtained results, let us consider an example of a $50$-vertex networked control system depicted in Figure~\ref{fig:50graphdom}. {The $50$-vertex graph is an Erdős–Rényi random undirected connected graph where an edge is included to connect two vertices with a probability of 0.5.}
Parameters of the system are selected as follows: $\theta_i = 0.5$, $\delta_i = 1,~\forall i \in \Vc$; the cost for the number of utilized sensors is set as $\cfr(|\Mc|) = \kappa |\Mc|$ where $\kappa = 5$; the beliefs of the defender and the malicious adversary in the location of the {target} vertex given an attack vertex are assumed to be uniformly distributed; and the sensor budget $n_s = 3$.
{It is worth noting that the mixed-integer SDP \eqref{mixedintegersdp} considers all possible pairs of $(a,\rho)  \in \Vc \times \Vc_{-\rho}$, yielding $50 \times 49 = 2450$ attack scenarios, which are considerable.}
\subsection{The Stackelberg optimal action}
First, we {begin with finding all the dominating sets of the considered $50$-vertex graph (see Figure~\ref{fig:50graphdom}).}
By investigating all the subsets $\Mc \subset \Vc$  where $|\Mc| \leq n_s$, twenty subsets satisfy the necessary and sufficient condition \eqref{dom_set_cond}, which are dominating sets. 
One of those dominating sets is illustrated in Figure~\ref{fig:50graphdom} where elements of a dominating set are coded blue.
{The computational time for finding dominating sets with the sensor budget $n_s = 3$ is under one second.}

{Next, we validate the obtained result of \textit{Theorem~\ref{th:dom_set_bound_impact}}.}
From Figure~\ref{fig:50graphdom}, let us consider a system $\Sigma_{m_k} \triangleq (-\bar L,e_a,e_{m_k}^\top,0)$ where $e_a$ represents the input at any vertex and $e_{m_k}$ represents the monitor output at a blue vertex.
We simply examine that there exists at least one blue vertex such that the relative degree of $\Sigma_{m_k}$ is never greater than one.
Thus, the cost for the defender and the expected worst-case impact of stealthy attacks are always bounded according to the result in \textit{Theorem~\ref{th:dom_set_bound_impact}}.
{To validate their boundedness, we compute the defense cost \eqref{game_payoff_def} and the expected worst-case impact of stealthy attacks \eqref{game_payoff_Q} for an arbitrary pair of a vertex $a \in \Vc$ and a dominating set $\Mc$.}
Through the computation, the maximum cost for the defender and the maximum expected worst-case impact of stealthy attacks are obtained as follows: $R(a,\Mc) \leq 50.2456$ and $Q(a,\Mc) \leq 48.4235$, which verifies the result in \textit{Theorem~\ref{th:dom_set_bound_impact}}.

Finally, 
{the best monitor set $\Mc^\star$ is found by directly solving \eqref{mixedintegersdp}.
}
The optimal action $\Mc^\star$ for the defender consists of three blue vertices in Figure~\ref{fig:50graphdom} that yields the minimum cost of $R(a^\star(\Mc^\star),\Mc^\star) = 49.7985$. Given such an optimal action $\Mc^\star$, {the corresponding attack vertex $a^\star(\Mc^\star)$ yields the maximum expected worst-case impact of stealthy attacks $Q(a^\star(\Mc^\star),\Mc^\star) = 47.9764$.}

\subsection{Computational complexity}
As discussed above, the $50$-vertex networked control system (see Figure~\ref{fig:50graphdom}) gives us twenty dominating sets where the sensor budget is three ($n_s = 3$). 
This number is extremely smaller than the number of subsets of the vertex set which has at most $n_s$ elements,
i.e., $S(50,3) = 20875$. 
{Solving the mixed-integer SDP \eqref{mixedintegersdp} normally employs the branch-and-bound algorithm. Adopting the concept of dominating sets considers 20 branches instead of 20875 branches, significantly saving computational resources.
On the other hand, if we go for the parallel computing suggested in \textit{Algorithm~\ref{al:parallel}}, we only need to request 20 computer cores for those dominating sets, which are suitable for many computer cluster platforms.
}
\begin{figure} [!t]
    \centering
    \includegraphics[width=0.5\textwidth]{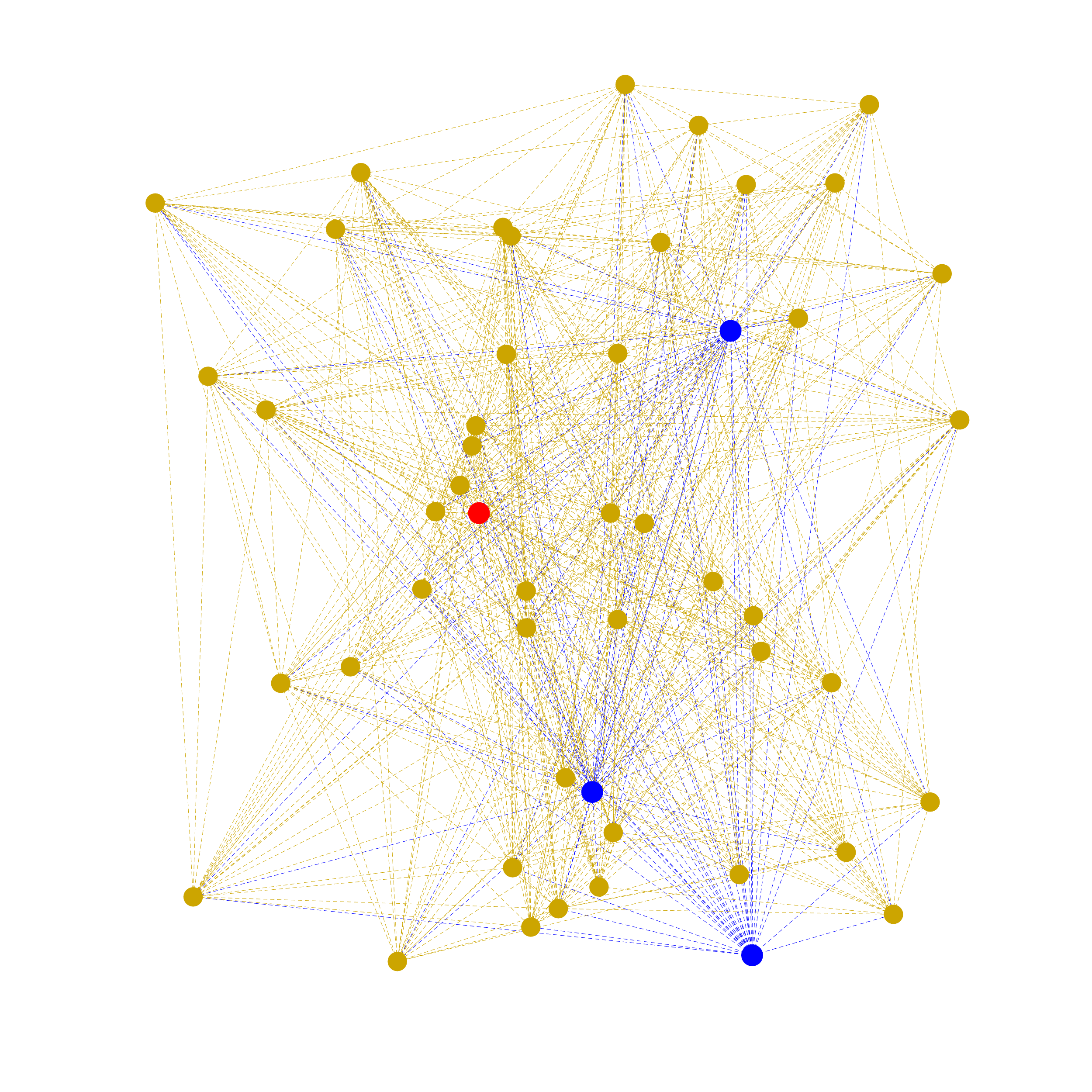}
    \caption{50-vertex graph where the optimal monitor vertices are coded blue and the optimal attack vertex is coded red.}
    \label{fig:50graphdom}
\end{figure}

\section{Conclusion}
\label{sec:concl}
In this paper, we investigated the security allocation problem in a networked control system when faced with a stealthy data injection attack.
The uncertain {target} vertex allowed us to formulate the objective functions of the defender and the adversary by considering probabilistic locations of the {malicious target}.
We presented a necessary and sufficient condition based on dominating sets under which the defender guarantees the boundedness of their cost and the expected worst-case impact of stealthy attacks.
The security allocation problem is cast in the Stackelberg game-theoretic framework where the defender plays the leader and the malicious adversary acts as the follower.
Then, we provided an algorithm to show the procedure of finding the Stackelberg optimal action.
The advantage of the proposed security allocation scheme was highlighted in the context of large-scale networks via a discussion on the computational burden and several numerical simulations.
{In future work, we can empower the adversary by allowing them to conduct attack signals on multiple vertices or sophisticated attacks such as multiplicative false data injection attacks and communication edge removal attacks. Further, we plan to characterize the Stackelberg optimal action for the defender and the adversary through graph properties such as centrality measures.
}
\appendices
\section{Proof of Lemma~\ref{lem:J_upbound}}
\label{app:J_upbound}
Showing \eqref{lem:Jtau1_upbound} is trivial when the monitor vertex set $\Mc$ has only one vertex. We assume that $\Mc$ has more than one monitor vertex. 
From the worst-case impact of stealthy attacks \eqref{Jtau1}, let us introduce the following optimization by removing $|\Mc|-1$ constraints except the constraint corresponding to a monitor vertex $m_k \in \Mc$ as follows:
		\begin{align}
			J_\rho(a,m_k) = \sup_{x(0)=0,~\zeta \in \Lc_{2e}} &~~ \norm{y_\rho}^2_{\Lc_{2}} \label{pf:Jtau_mk} \\
			\text{s.t.}~~~~~&~~ 
			\norm{y_{m_k}}^2_{\Lc_{2}} \leq \delta_{m_k}. \non 
		\end{align}
		The design of the optimization problem \eqref{pf:Jtau_mk} tells us that its feasible set contains the feasible set of the optimization problem \eqref{Jtau1}. Further, the two optimization problems \eqref{Jtau1} and \eqref{pf:Jtau_mk} have the same objective function. This implies that $J_\rho(a,\Mc) \leq J_\rho(a,m_k)$ for all $m_k \in \Mc$, directly resulting in \eqref{lem:Jtau1_upbound}.

\section{Proof of Lemma~\ref{lem:no_un_zero}}
\label{app:lem_pf_no_un_zero}
    Let us denote a tuple $(\bar \lambda_{m_k},\bar{x}_{m_k},\bar g_{m_k}) \in \Cbb \times \Cbb^N \times \Cbb$ as a zero dynamics of $\Sigma_{m_k}$, where a finite $\bar \lambda_{m_k}$ is called a finite invariant zero of $\Sigma_{m_k}$.
    From \textit{Definition~\ref{def:invariant_zero}}, one has that the tuple $(\bar \lambda_{m_k},\bar{x}_{m_k},\bar g_{m_k})$ satisfies
    \begin{align}
         \ba{cc}
        \bar \lambda_{m_k}  I + \bar L  & -e_a \\
        e_{m_k}^\top & 0
        \ea 
        \ba{c} \bar{x}_{m_k} \\ \bar g_{m_k} \ea 
        = 
        \ba{c} 0 \\ 0 \ea.
        \label{pen_mtr_lam_m}
    \end{align}
    The above equation is rewritten as
    \begin{align}
        \ba{cc}
        (\bar \lambda_{m_k} - \theta_0) I + \bar L + \theta_0 I & -e_a \\
        e_{m_k}^\top & 0
        \ea 
        \ba{c} \bar{x}_{m_k} \\ \bar g_{m_k} \ea 
        = 
        \ba{c} 0 \\ 0 \ea,
        \label{pen_mtr_lam_m1}
    \end{align}
    where $\theta_0 \in \Rbb_+$ is a uniform offset self-loop control gain.
    From \eqref{pen_mtr_lam_m1}, the finite value $(\bar \lambda_{m_k}-\theta_0) \in \Cbb$ is an invariant zero of a new state-space model $\tilde \Sigma_{m_k} \triangleq (-\bar L-\theta_0I,e_a,e^\top_{m_k},0)$.
    For all $\bar \lambda_{m_k} \in \Cbb$ satisfying \eqref{pen_mtr_lam_m1},
    the control gain $\theta_0$ can be adjusted such that $\theta_0 >$ Re$(\bar \lambda_{m_k})$, resulting in that $\tilde \Sigma_{m_k}$ has no finite unstable zero.
    Then, the self-loop control gains $\theta_i, ~ i \in \{1,2,\ldots,N\},$ in \eqref{sys:u} are tuned with $\theta_0$ such that the system $\Sigma_{m_k}$  is identical with $\tilde \Sigma_{m_k}$.
    By this tuning procedure, the system $\Sigma_{m_k}$ also has no finite unstable invariant zero.

\section{Proof of Theorem~\ref{th:red_con}}
\label{app:th_pf_red_con}
    The result in \textit{Lemma~\ref{lem:inv_zero}} enables us to investigate invariant zeros of the systems $\Sigma_\rho$ and $\Sigma_{m_k},~\forall m_k \in \Mc$.
	Based on \textit{Lemma \ref{lem:no_un_zero}}, $\Sigma_{m_k}$ has no finite unstable invariant zero, which leaves us to analyze infinite invariant zeros of those systems.
    Recall the equivalence between the relative degree of a SISO system and the degree of its infinite zero (see \textit{Remark~\ref{rem:re_deg_tf}}), a necessary condition to guarantee the {boundedness} of the optimization problem \eqref{lem:Jtau1_upbound_def} is that there exists at least one system $\Sigma_{m_k} (m_k \in \Mc)$ such that 
    the number of its infinite invariant zeros is not greater than that of the system $\Sigma_\rho$.
 This implies  $r_{(m_k,a)} \leq r_{(\rho ,a)}$. 
 For sufficiency, it remains to show that 
 	if $r_{(m_k,a)} \leq r_{(\rho ,a)}$, all the infinite zeros of the system $\Sigma_{m_k}$ are also infinite zeros of the system $\Sigma_{\rho}$. The following proof is adapted from our previous results in \cite[Th. 7]{nguyen2022single}.
In the investigation, we make use of the definition of infinite invariant zeros in \cite[Def. 2.4]{morris2010invariant}.
 	We  investigate infinite zeros of $\Sigma_{m_k}$ and $\Sigma_{\rho}$
    by starting from
	their transfer functions with zero initial states
	\begin{align}
		G_{(\rho,a)}(s) &= e_\rho^\top (sI + \bar L)^{-1} e_a 
		= \frac{P_{(\rho,a)}(s)}{Q(s)}
		,\non \\
		G_{(m_k,a)}(s) &= e_{m_k}^\top (sI + \bar L)^{-1} e_a 
		= \frac{P_{(m_k, a)}(s)}{Q(s)}
		,
		\label{trans_fcn}
	\end{align}	
	where $s \in \Cbb$ is the Laplace complex variable.
	Based on \textit{Remark \ref{rem:re_deg_tf}}, it gives that
	$P_{(\rho,a)}(s)$, $P_{(m_k,a)}(s)$, and $Q(s)$ are the polynomials of  degrees $N-r_{(\rho,a)}$, $N-r_{(m_k,a)}$, and $N$, respectively.
	Let us denote $z_\tau = \sigma_\tau + j \omega_\tau \in \Cbb,~ \tau \in \{1,2,\ldots,r_{(m_k,a)} \}$ with infinite module as infinite invariant zeros of $\Sigma_{m_k}$. 
	Indeed, the zero $z_\tau~(1 \leq \tau \leq r_{(m_k,a)})$ is an infinite invariant zero of maximal degree $r_{(m_k,a)}$ of the system $\Sigma_{m_k}$ \cite[Def. 2.4]{morris2010invariant} if it satisfies
	\begin{align}
	    &\lim_{\norm{z_\tau} \ra \infty} z_\tau^q G_{(m_k,a)}(z_\tau) = 0, ~ (0 \leq q \leq r_{(m_k,a)} - 1), 
	    \non \\
	    &\lim_{\norm{z_\tau} \ra \infty} z_\tau^{r_{(m_k,a)}} G_{(m_k,a)}(z_\tau) \neq 0.
	\end{align}
	Further, with $0 \leq q \leq r_{(m_k,a)} - 1$, we also basically have
	\begin{align}
	    &\lim_{\norm{z_\tau} \ra \infty} z_\tau^q G_{(\rho,a)}(z_\tau) = 
        \lim_{\norm{z_\tau} \ra \infty}  \frac{z_\tau^q P_{(\rho,a)}(z_\tau)}{Q(z_\tau)} = 0. \label{Pta_inf_zero_1}
	\end{align}
The above limit \eqref{Pta_inf_zero_1} holds because the denominator $z_\tau^{q}P_{(\rho,a)}(z_\tau)$ is the polynomial of degree $N-r_{(\rho,a)}+q \leq N-1 < N$, where $N$ is  the degree of the polynomial $Q(z_\tau)$.
This implies that any infinite zeros $z_\tau$ of maximal degree $r_{(m_k,a)}$ of the system $\Sigma_{m_k}$ are also infinite zeros of degree $r_{(m_k,a)}$ of the system $\Sigma_\rho$. 

\section{Proof of Lemma~\ref{lem:dominating_set}}
\label{app:lem_pf_dominating_set}
Let us decompose $\Cc(\Mc) = \Cc_A(\Mc) + \Cc_I(\Mc)$ where $\Cc_A(\Mc) = \sum_{{m_k} \in \Mc} A e_{m_k}$ and $\Cc_I(\Mc) = \sum_{{m_k} \in \Mc} e_{m_k}$. Entry $i$-th of $\Cc_I(\Mc)$ takes  $0$ if vertex $i$ does not belong to $\Mc$ and $1$ if vertex $i$ belongs to $\Mc$.
Entry $i$-th of $\Cc_A(\Mc)$ takes $0$ if all the neighbors of vertex $i$ do not belong to $\Mc$ and a non-zero value if at least one neighbor of vertex $i$ belongs to $\Mc$.
Thus, entry $i$-th of $\Cc(\Mc)$ takes $0$ if vertex $i$ and all of its neighbors do not belong to $\Mc$; takes a non-zero value if vertex $i$ or one of its neighbors belong to $\Mc$.
If the condition \eqref{dom_set_cond} fulfills, the vector $\Cc(\Mc)$ has no zero entry.
This implies that an arbitrary vertex in $\Vc$ is either a vertex of $\Mc$ or a neighbor of a vertex of $\Mc$, resulting in that $\Mc$ is a dominating set.
\section{Proof of Proposition~\ref{proposition:bigLMI}}
\label{app:propositionbigLMI}
{
The proof is constructed by the following three main steps. The first step modifies \eqref{opt_LM_security_metric} such that monitor vertices are represented by new binary variables. 
In the second step, the optimization problem \eqref{Eq_M_star} is rewritten by using \eqref{game_payoff_def}-\eqref{game_payoff_Q} and \eqref{Eq_a_star} to expose the presence of \eqref{opt_LM_security_metric}.
In the last step, the mixed-integer SDP problem \eqref{mixedintegersdp} is obtained by replacing the modified version of \eqref{opt_LM_security_metric} into the rewritten version of \eqref{Eq_M_star}, incorporating the proposed concept of dominating sets.
}

{ \textit{Step 1}:
Due to the sensor budget $n_s$, a monitor set is restricted as a selection from the collection of all the subsets of the vertex set $\Vc$ with a cardinality of at most $n_s$, which is denoted as $\Sbb$. 
Define an $N$-dimensional binary vector $z_i$ as a representation of the $i$-th subset in the collection $\Sbb$ where $j$-th entry of $z_i$ being equal to 1 indicates that vertex $j$ belongs to the $i$-th subset.
The collection $\Sbb$ can thus be compactly represented as the following binary matrix:
\begin{align}
    Z_{\Sbb} = \left[ z_1, z_2, \ldots, z_{|{\Sbb}|} \right]. \label{app:ZSbb}
\end{align}
Define a $|{\Sbb}|$-dimensional binary vector $v$ as a selection from the collection $\Sbb$ where $j$-th entry of $v$ being equal to 1 indicates that $j$-th subset in $\Sbb$ is chosen. Therefore, a monitor set $\Mc$, which is chosen by the defender, can be represented as follows:
\begin{align}
    \sum_{m \, \in \, \Mc} e_m = Z_{\Sbb} \, v,~~ \textbf{1}_{|\Sbb|}^\top v = 1. \label{app:zv}
\end{align}
For a given pair $(a,\rho)$, define an $N$-dimensional non-negative vector $\bar z_{(a,\rho)} \in \Rbb^N$ as a replacement of the variable $\gamma_{m_k}$ in \eqref{opt_LM_security_metric} such that $i$-th entry of $\bar z_{(a,\rho)}$ is greater than zero if vertex $i$ belongs to the chosen monitor set $\Mc$, otherwise it is zero. 
The representation of the chosen monitor set $\Mc$ in \eqref{app:zv} enables us to formulate the above-defined constraint of the new variable $\bar z_{(a,\rho)}$  as follows:
\begin{align}
    0 \leq \bar z_{(a,\rho)} \leq \tilde M \, Z_{\Sbb} \, v,
\end{align}
where $\tilde M$ is a large positive number, also called a big ``M'' \cite{milovsevic2023strategic}. It is worth noting that the structure of the new variable $\bar z_{(a,\rho)}$ enables us to rewrite the objective function of \eqref{opt_LM_security_metric} and the last term of its constraint as follows:
\begin{align}
    \delta^\top \, \bar z_{(a,\rho)} &= \sum_{m_k \in \Mc} \delta_{m_k} \, \gamma_{m_k}, \label{app:objz}\\
    \textbf{diag} \left(\ba{c}
	\bar z_{(a,\rho)} \\ 0
	\ea \right) &= \sum_{m_k \in \Mc} \gamma_{m_k} \ba{c}
	e_{m_k} \\ 0
	\ea 
	\ba{cc}
	e_{m_k}^\top & 0
	\ea, \non
\end{align}
where $\delta = [\delta_1,\delta_2,\ldots,\delta_N]^\top$ is the vector of all the alarm thresholds.
By invoking \eqref{app:ZSbb}-\eqref{app:objz}, the output-to-output gain security metric \eqref{opt_LM_security_metric} for a given pair $(a,\rho)$ can be rewritten as follows:
\begin{align}
    J_\rho(a,\Mc) =&  \min_{\bar z_{(a,\rho)}  \in  \Rbb^N, \, v \, \in \, \{0,1\}^{|\Sbb|}, \, P_{(a,\rho)} \, \in \, \Rbb^{N \times N}} ~ \delta^\top \bar z_{(a,\rho)} \label{app:biglmi} \\
    \text{s.t.} ~
    & \textbf{1}_{|\Sbb|}^\top \, v = 1,~  0 \leq \bar z_{(a,\rho)} \leq \tilde M \, Z_{\Sbb} \, v, 
    \non \\
    & \hspace{-1cm}
    \ba{cc}
	-\bar L P_{(a,\rho)} - P_{(a,\rho)}  \bar L & P_{(a,\rho)}  e_a \\
	e_a^\top P_{(a,\rho)} & 0
	\ea  + \textbf{diag} \left(\ba{c}
	e_\rho \\ 0
	\ea \right)
    \non \\
    & 
    - \textbf{diag} \left(\ba{c}
	\bar z_{(a,\rho)} \\ 0
	\ea \right) \leq 0,~ P_{(a,\rho)} = P_{(a,\rho)}^\top \geq 0. \non
\end{align}
}

{\textit{Step 2:}
On the other hand,
from \eqref{game_payoff_def}-\eqref{game_payoff_Q} and \eqref{Eq_M_star}, the optimization problem $\min_{\Mc \, \in \, \Dbb} \, R(a^\star(\Mc),\Mc)$ in \eqref{Eq_a_star} can be rewritten equivalently as:
\begin{align}
    \min_{\Mc \, \in \, \Dbb, \,\beta \, > \, 0} ~& \bigg[ ~c(|\Mc|) + \beta ~\bigg] \label{app:Ropt}
    \\
    \text{s.t.}~~~&
    \sum_{\rho \in \Vc_{-a}} \pi_a(\rho) J_\rho (a,\Mc) ~\leq ~\beta, ~~ \forall \, a \, \in \, \Vc.
    \non
\end{align}
}

{\textit{Step 3:}
Intuitively, the boundedness of \eqref{app:Ropt} is guaranteed by the boundedness of \eqref{app:biglmi} $\forall \, a \in \Vc$.
Fortunately, the result of \textit{Theorem~\ref{th:dom_set_bound_impact}} implies that \eqref{app:biglmi} is bounded $\forall \, a \in \Vc$ if, and only if, the monitor set $\Mc$ is a dominating set. Therefore, the collection $\Sbb$ in \eqref{app:biglmi} can be restricted into the collection $\Dbb$ of all the dominating sets. This restriction and \eqref{app:biglmi}-\eqref{app:Ropt} gives us \eqref{mixedintegersdp}.
}

\bibliographystyle{IEEEtran}
\bibliography{mybibfile}

\end{document}